\documentclass[preprint,review,12pt]{elsarticle}

\usepackage{upgreek} 
\usepackage{amsmath} 
\usepackage{mathtools} 
\usepackage{amssymb} 
\usepackage{amsfonts} 
\usepackage{dutchcal}
\usepackage{bbold}
\usepackage{siunitx}
\usepackage{subfigure}
\usepackage{algorithm}
\usepackage{booktabs}
\usepackage{algpseudocode}
\usepackage{hyperref}
\usepackage{makecell}

\newtheorem{assumption}{Assumption}
\newtheorem{remark}{Remark}
\newtheorem{definition}{Definition}
\newtheorem{lemma}{Lemma}
\newtheorem{proof}{Proof}
\newtheorem{corollary}{Corollary}

\usepackage{graphicx}

\newcommand{\nint}[1]{\ensuremath\left\lfloor#1\right\rceil}

\newcommand{\ui}[2]{#1_{\text{#2}}}

\newcommand{\pnorm}{n}
\newcommand{\pdegree}{\kappa}

\DeclareMathOperator{\interior}{int}
\newcommand{\mathcalE}{\mathcal{E}} 
\newcommand{\mathcalD}{\mathcal{D}} 

\usepackage{tikz}
\usepackage{pgfkeys}
\usepackage{calc}
\usepackage{pgfplots}
\pgfplotsset{compat=1.17}
\usetikzlibrary{arrows}
\usetikzlibrary{calc}
\usepgflibrary{arrows}
\usetikzlibrary{positioning}
\usetikzlibrary{intersections}
\usetikzlibrary{shapes.geometric}
\usetikzlibrary{decorations.pathreplacing}
\usetikzlibrary{patterns,decorations.pathmorphing,decorations.markings}
\usetikzlibrary{external}
\usetikzlibrary{plotmarks}
\usetikzlibrary{arrows.meta}
\usepgfplotslibrary{patchplots}
\usepackage{grffile}
\usepgfplotslibrary{fillbetween}
\usepackage{xfrac}

\newcommand{\includetikz}[1]{%
	\tikzifexternalizing{%
		\def\DOIT{1}%
	}{%
		\IfFileExists{#1.pdf}{%
			\includegraphics[scale=1]{#1.pdf}%
			\def\DOIT{0}%
		}{%
			\def\DOIT{1}%
		}%
	}%
	\if1\DOIT
	\tikzsetnextfilename{#1}
	\input{#1.tikz}
	\fi
}

\newlength\figureheight
\newlength\figurewidth

\makeatletter
\def\ps@pprintTitle{%
	\let\@oddhead\@empty
	\let\@evenhead\@empty
	\def\@oddfoot{}
	\let\@evenfoot\@oddfoot}
\makeatother

\journal{European Journal of Process Control}

\begin{document}
	\begin{frontmatter}
		\title{Safety and Security: Experimental Validation of~Encrypted Model Predictive Control}
		
		\author[STU]{Juraj Holaza}\ead{juraj.holaza@stuba.sk}
		\author[STU]{Martin Kal\'{u}z}\ead{martin.kaluz@stuba.sk}
		\author[STU]{Mat\'{u}\v{s} Furka}\ead{matus.furka@stuba.sk}
		\author[CTU]{Martin Klau\v{c}o}\ead{martin.klauco@cvut.cz}
		\author[STU]{Juraj~Oravec}\ead{juraj.oravec@stuba.sk}

		\address[STU]{Slovak University of Technology in Bratislava, Faculty of Chemical and Food Technology, Institute of Information Engineering, Automation, and Mathematics, Radlinsk\'{e}ho 9, SK812-37 Bratislava, Slovak Republic}
		\address[CTU]{Department of Control Engineering, Czech Technical University in Prague, Czech~Republic}

		\begin{abstract}
			In this paper, we revisit the problem of an encrypted model predictive control (MPC) design, representing a significant challenge in the recent field of secure process control. Existing methods in secure optimization-based control are non-existent and even partial implementation fails to address the closed-loop system stability and recursive feasibility properties of the constrained MPC. To overcome these limitations, we propose a novel approach that utilizes a polynomial approximation of the optimal control law. This method evaluates the explicit control law within a fully homomorphic encryption framework, ensuring that the controller is securely deployed on any third-party or cloud-based platform, with both process data and controller coefficients protected. Experimental results from a laboratory-scale implementation and validation of the proposed privacy-aware control method demonstrate its advantages.
		\end{abstract}
		
	\end{frontmatter}
	
\section{Introduction}
\label{sec:intro}

Cryptographic security and data privacy plays a vital role in implementing control strategies in modern internet-based cloud services, e.g., see~\cite{teranishi:2024:tcns}, or~\cite{Darup2021} and references therein. Having open, unsecured lines of communication between parts of the closed-loop control systems (sensor, control law evaluator, actuator), considerably restricts the possibility of deploying evaluation of control laws to third-party platforms that just evaluate the control law and provide values for the actuator. Such an arrangement of closed-loop control systems is gaining popularity while it is being demanded by the Internet of Things outreach, which requires everything to be available in the cloud~\cite{7165580}. Unfortunately, the control of industrial processes must remain private not only in the sense of {data transition} but control strategies employed by industries must be carefully guarded since total economic profit heavily depends on the control performance. The control strategy not only affects the profits but is directly responsible for the stability and technological security of the industrial processes itself~\cite{KNAPP201541}. Relevancy of the cybersecurity in the control domain is also emphasized in the recent studies~\cite{lin:2022:ijrnc} or~\cite{chen:2020:ijrnc}.

The secure control approach presented in this paper targets scenarios, where the following challenges are of paramount importance:
\begin{itemize}
	\item \textit{total data privacy}: the platform evaluating the control law will not have access to physical quantities,
	\item \textit{total controller privacy}: the platform evaluating the control law will not have access to the controller coefficient or particular structure,
	\item \textit{close-to-optimal control laws}: the technological constraints are enforced, and
	\item deployment of the control law on \textit{any platform}: either embedded or cloud-based.
\end{itemize}
To address the first two challenges, a fully homomorphic framework is used to evaluate the control law based on encrypted process measurements~\cite{BFV2012}. 
Numerous works in the domain of secure process control have already been published 
considering state feedback~\cite{shen:2021:smcs,stobbe:2022:cdc,BFV_stateControl2024}, polynomial~\cite{Darup2020,MoritzDarup_2025_polyEcryption} or MPC~\cite{Feng2023,EncryptedRobustMPC2026,EncryptedDistributedMPC2024} controllers, decentralized optimization~\cite{Huo2023} or formation control~\cite{Marcantoni2022}. However, the majority of the recent works rely on the use of 
Paillier cryptosystem~\cite{Paillier1999}, which is a partially homomorphic scheme. Then, one needs to compromise by either providing secure control scenarios with \textit{secured data} and a \textit{public controller} or vice versa. None of the previously mentioned approaches, however, solves the problem of implementing a constraint-enabled control with a fully homomorphic encryption framework, see excellent surveys on the progress in the encrypted control in~\cite{SB23,NonlinearControl_Survey_2026}.

Early results on how the challenges $3$ and $4$ of the above list are achieved were published in~\cite{Darup2018}. Authors construct an explicit model predictive control, but since the associated point location problem can not be performed efficiently in a homomorphic fashion, the structure of regions must be either revealed to the non-colluding third-party server or overload the evaluation at the controlled plant side. 
This obstacle was later addressed, e.g., in~\cite{SD20} by introducing a pair of trustworthy servers. 
The convex reformulation of the explicit control law results in a significantly increased number of critical regions to be evaluated leading to the 
overloaded point location problem. 
The implicit (non-explicit) encrypted MPC was designed in~\cite{Alexandru2018}, where either the controlled plant-side or the auxiliary non-colluding 
server assists in evaluating the MPC. 
The privacy-preserving explicit MPC based on the tokenized affine transformations of the control law was designed in~\cite{HS22}. This approach requires a computationally intensive construction of the virtual subregions to the corresponding original critical region of the explicit solution map. 


The paper is organized as follows: The Section~\ref{sec:he:prelim} presents theoretical backgrounds on encrypted controller design and the problem statement. 
Section~\ref{sec:mpc} describes the controller design procedure, including the robust MPC formulation, its explicit representation, and the resulting polynomial approximation. Experimental results obtained on a laboratory setup are presented in Section~\ref{sec:CaseStudy}.
Concluding summary is given in Section~\ref{sec:Conclusion}.
%


\section{Preliminaries and problem statement}
\label{sec:he:prelim}

\subsection{Encrypted Control}
\label{sec:he}

The encryption schemes transform values of manipulated and process variables into plaintexts $p$ that are further encrypted into ciphertexts $c$. Each transformation is formally defined by a separate function, which has its inverse operation, as given in
\begin{subequations}
	\begin{align}
		m &\xrightarrow{f^{\mathcal{E}}(\cdot)} p \xrightarrow{\mathcal{E}(\cdot)} c \label{eq:enc} \, , \\
		c &\xrightarrow{\mathcal{D}(\cdot)} p \xrightarrow{\ui{f}{D}(\cdot)} m \label{eq:dec} \, .
	\end{align}
\end{subequations}
HE algorithms take a special place in cryptography since they provide the possibility of performing mathematical operations over ciphertexts. The property of \textit{homomorphism} ensures that the evaluator of the mathematical operation, i.e., the platform where the control law is evaluated, does not possess any knowledge about manipulated or process variables. The fundamental relation defining the base upon all the HE schemes are built on is defined as
\begin{equation} 
	\mathcal{D}(\mathcal{E}(f^{\mathcal{E}}(m_1)) \odot \mathcal{E}(f^{\mathcal{E}}(m_2))) = m_1 {\boxplus} m_2 \, ,
\end{equation}   
where $m_1$ and $m_2$ are original messages and $\odot$ is mathematical operation over encrypted $m_1$ and $m_2$ resulting in ${\boxplus}$ over original messages. Each HE scheme has predefined its homomorphic properties represented by the pairs of operations $\odot$ and ${\boxplus}$.

In HE schemes we define two mathematical operations that allow us to perform calculations over encrypted messages, specifically, we consider: 
\begin{itemize}
	\item \textit{homomorphic addition} for which has each cryptographic scheme predefined mathematical relations to be carried out over ciphertexts resulting in {addition of} the two or more original messages. Homomorphic addition is considered as a \textit{cheap} operation in terms of computational complexity when computing over ciphertexts.
	\item \textit{homomorphic multiplication}, which in comparison to addition is an \textit{expensive} operation, since the resulting ciphertexts usually increase their size. Thus, there is a requirement for the presence of additional algorithms used to handle the size of the resulting ciphertext. This intuitively leads to an increase in computational complexity.
\end{itemize}

Given the possibility of performing one or both homomorphic operations, we divide HE schemes into:
\begin{itemize}
	\item \textit{partially homomorphic} (PHE) framework enabling to perform homomorphic addition \emph{or} multiplication
	\item \textit{fully homomorphic} (FHE) framework performing both homomorphic addition and homomorphic multiplication
\end{itemize}

\begin{figure}
	\centering
	\includegraphics[width=0.8\textwidth]{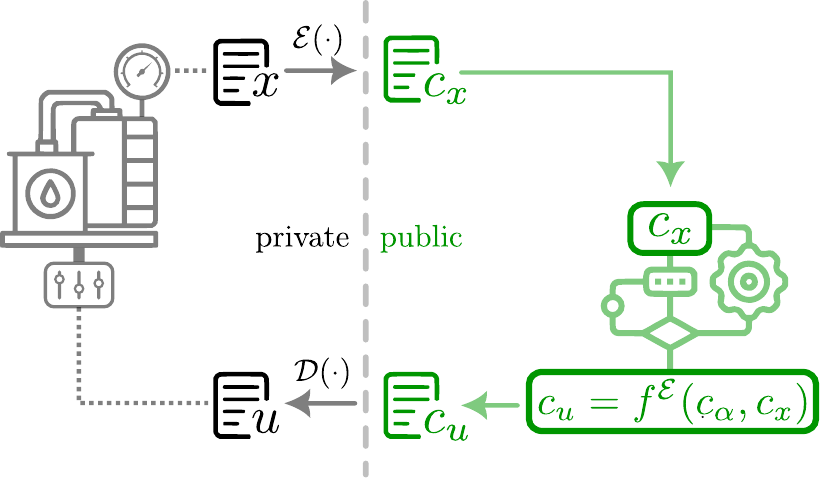}
	\caption{{Schematic representation of the closed-loop control system with fully homomorphic implementation. The gray part denotes the controlled process (\textit{private zone}), while the green part of the picture depicts the cloud-based evaluation of encrypted control law (\textit{public zone}).}}
	\label{fig:general_HE_schemev2}
\end{figure}

In this paper, we consider the implementation of the later-mentioned FHE scheme to handle the evaluation of the polynomial control law in a fully homomorphic way. {By allowing both multiplication and addition operations with encrypted values, we can fully outsource the evaluation of the control law to cloud services. Such an implementation of an FHE-based controller in the closed-loop is depicted in Figure~\ref{fig:general_HE_schemev2}.}

\subsection{Fully Homomorphic Encryption Framework}
\label{sec:hea:bfv}

In this paper we consider the \textit{Brakerski/Fan-Vercauteren} (BFV) cryptosystem~\cite{BFV2012}, a member of the FHE frameworks family.
The BFV scheme is similar to the \textit{Brakerski-Gentry-Vaikuntanathan} (BGV) scheme~\cite{BGV2014}, but employs different noise management techniques~\cite{kim2021}. Both BGV and BFV are based on the hardness of the \textit{Ring Learning With Errors} (RLWE) problem~\cite{RLWE}, which provides strong security guarantees against quantum and classical attacks. In BFV, plaintexts and ciphertexts are represented as polynomials over two distinct rings
\begin{subequations}
	\begin{align}
		R_{\tau} &= \mathbb{Z}_{\tau} [X]/(X^{\mathcal{N}} + 1) , \label{eq:plain:space} \\
		R_\mathcal{q} &= \mathbb{Z}_{\mathcal{q}} [X]/(X^{\mathcal{N}} + 1) , \label{eq:cipher:space}
	\end{align}
\end{subequations}
where, respectively,~\eqref{eq:plain:space} is a space for plaintext polynomials with coefficients smaller than plaintext modulus $\tau$, and~\eqref{eq:cipher:space} represents polynomial space for
each coordinate (polynomial) of the ciphertext with coefficients smaller than coefficient modulus $\mathcal{q}$. Both coefficient moduli ensure that the plaintext and ciphertext operations are performed in the finite fields $\mathbb{Z}_{\tau}$ and $\mathbb{Z}_{\mathcal{q}}$, respectively. The size of rings~\eqref{eq:plain:space} and~\eqref{eq:cipher:space} is defined by polynomial modulus $(X^{\mathcal{N}} + 1)$, a cyclotomic polynomial characterized by a polynomial modulus degree $\mathcal{N}$.

Since the BFV is built on the RLWE problem, the ciphertext is represented by a tuple of two polynomials, i.e., $c=(c_{(1)},c_{(2)})$. Thus the polynomial space for ciphertext is defined as $R_\mathcal{q}^2$, where $2$ stands for the tuple of two elements. BFV provides an evaluation of both homomorphic additions and multiplications with exact results. The cryptographic scheme within the homomorphic properties of BFV is described in detail in~\cite{BFV2012}.

Since BFV is based on integer numbers, we consider a simple quantization as follows
\begin{equation}
	\label{eq:quant}
	\ui{m}{I} = \nint{ m \cdot 10^{\theta} } ,
\end{equation}
where $\ui{m}{I}$ is integer form of real-valued message $m$ regarding the precision degree $\theta$. {In our proposal, we consider encrypted closed-loop as depicted in Fig.~\ref{fig:general_HE_schemev2}. In this figure, the states $x$ and control inputs $u$ are the real-valued messages $m$ encrypted to ciphertexts of states $\ui{c}{x}$ and inputs $\ui{c}{u}$. The ciphertext of designed polynomial controller is noted as $c_{\alpha}$, while the evaluation of the control law over ciphertexts is denoted by function $f^{\mathcal{E}}(\cdot)$. }

\subsection{Problem Statement}
\label{sec:problem_statement}

The main aim of this paper is to design a secure, robust, stable, and easy to implement MPC policy. 
The security is imposed by using the fully encrypted evaluation of control law, which however induces quantization of control inputs and system states, see~\eqref{eq:quant}. 
To address this quantization effect one can design a robust MPC policy. Literature provides various methods such as the min-max-like robust MPC~\cite[Chapter 15]{borrelli_bemporad_morari_2017} or the tube MPC~\cite{Darup2021} designs to name a few. In this paper we apply the first mentioned approach.

In what follows, we point out that the inherited quantization error~\eqref{eq:quant} has to be considered in the controller design procedure. Specifically, we show that it can be formulated as an additive disturbance in the prediction model, leading to a robust MPC design.

\section{Encrypted Model Predictive Controller via Polynomial Approximation}
\label{sec:mpc}

In this section, we show how to approximate the resulting MPC policy by a polynomial feedback law to reduce computational requirements and its memory footprint, hence making the implementation much easier. It should be empathized that all the aforementioned approaches are performed such that the feasibility and stability properties of the original MPC policy are still preserved. The subsequent steps of encrypted approximated MPC design are depicted at Figure~\ref{EncMpcDesing}.

\begin{figure}[h]
	\centering
	\includegraphics[scale=0.5]{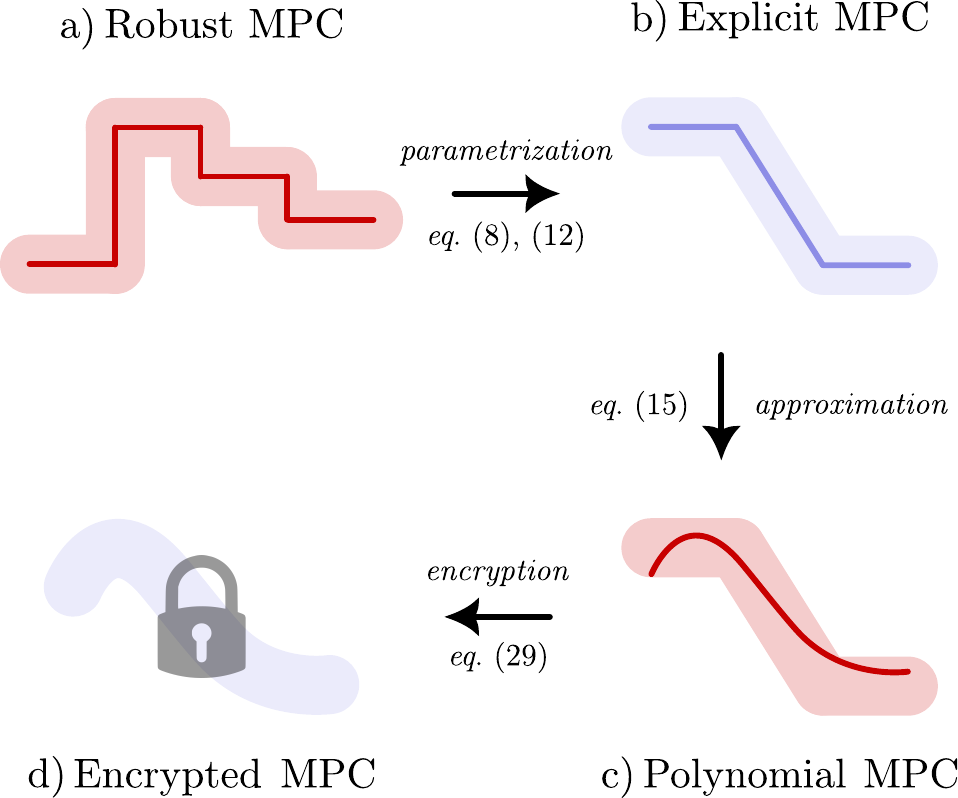}  
	\caption{{Encrypted MPC design flowchart. The designed robust MPC $a$) is parametrized, which results in explicit control law $b$) respecting the constraints and uncertainties. The explicit solution is then approximated by polynomial resulting in $c$). The individual coefficients of this polynomial are encrypted and form the secured form of the approximated controller $d$). Our main contribution is emphasized by red color.}}
	\label{EncMpcDesing}
\end{figure}

\subsection{Robust Encrypted MPC Design}
\label{sec:mpc:theory}

First, for MPC design purposes, we need a reliable mathematical model predicting the future behavior of the system dynamics. Consider a widely-used linear time-invariant (LTI) system in discrete time domain affected by the interval uncertainties $f(x(t),u(t),q_\text{x}(t),q_\text{u}(t),d(t))$ defined as
\begin{equation}
	\label{eq:mpc:theory:sys2}
	f(\cdot) = A (x(t) + q_\text{x}(t)) + B (u(t)+q_\text{u}(t)) + D d(t),
\end{equation}
with states $x(t) \in \mathbb{R}^{n_\text{x}}$, inputs $u(t) \in \mathbb{R}^{n_\text{u}}$, state quantization error $q_\text{x} \in \mathbb{Q}_{\text{x}} \subset \mathbb{R}^{n_\text{x}}$, input quantization error $q_\text{u}(t) \in \mathbb{Q}_{\text{u}} \subset \mathbb{R}^{n_\text{u}}$, and process disturbance $d(t) \in \mathfrak{D} \subset \mathbb{R}^{n_\text{d}}$. We note that both quantization errors $q_\text{x}(t)$ and $q_\text{u}(t)$ are within compact polyhedral sets  $\mathbb{Q}_{\text{x}}$ and $\mathbb{Q}_{\text{u}}$, respectively. Note, the origin of these errors arises from the encoding in~\eqref{eq:quant}, where quantization of real numbers into integers is required. Hence, the prediction model  in~\eqref{eq:mpc:theory:sys2} represents an LTI system under additive uncertainty.

Let us rewrite~\eqref{eq:mpc:theory:sys2} into its equivalent compact form
\begin{equation}
	\label{eq:mpc:theory:sys}
	\widetilde{f}(x(t),u(t),w(t)) = A x(t) + B u(t) + W w(t),
\end{equation}
where $w(t) = \begin{bmatrix}d^\top(t) & q_\text{x}^\top(t) & q_\text{u}^\top(t) \end{bmatrix}^\top \in \mathbb{R}^{n_\text{d}+n_\text{x}+n_\text{u}}$ and $W = \begin{bmatrix}	D & A & B \end{bmatrix}$.
Assume that the system~\eqref{eq:mpc:theory:sys} is subjected to non-empty polyhedral sets
\begin{equation}
	\label{eq:mpc:theory:cons}
	x(t) \in \mathcal{X}, \ u(t) \in \mathcal{U}, \ w(t) \in \mathcal{W},
\end{equation}
where $\mathcal{X} \subseteq \mathbb{R}^{n_\text{x}}$, $\mathcal{U} \subseteq \mathbb{R}^{n_\text{u}}$, and $w(t) \in \mathcal{W} \subset \mathbb{R}^{n_\text{d}+n_\text{x}+n_\text{u}}$.

We define an MPC problem to drive the system~\eqref{eq:mpc:theory:sys} to the origin while fulfilling constraints in~\eqref{eq:mpc:theory:cons} as
\begin{subequations}
	\label{eq:mpc:theory:MPC}
	\begin{align}
		J^\star = \min_{\substack{U_{\text{N}}}} \;& \Vert x_{N} \Vert_{Q_{\mathrm{N}}}^\pnorm + 
		\sum_{k=0}^{N-1} \left( \Vert x_{k} \Vert_{Q_{\mathrm{x}}}^\pnorm + \Vert u_{k} \Vert_{Q_{\mathrm{u}}}^\pnorm \right) \label{eq:mpc:theory:MPC_obj}\\
		\text{s.t.} \ 	&x_{k+1} = \widetilde{f}(x_k,u_k,w), 	\label{eq:mpc:theory:MPC_model}\\
		&x_k \in \mathcal{X}, 	u_k \in \mathcal{U},  w_k \in \mathcal{W},  \label{eq:mpc:theory:MPC_cons}\\
		&x_N \in \mathcal{X}_{\mathrm{N}}, \label{eq:mpc:theory:MPC_terminalcon}\\
		&x_0(t) = x(t), \label{eq:mpc:theory:MPC_initcon}
	\end{align}
\end{subequations}
where~\eqref{eq:mpc:theory:MPC_model}--\eqref{eq:mpc:theory:MPC_cons} are enforced for the entire prediction horizon $N$, i.e., $k = 0, \dots, N-1$. Next, $x_k$, $u_k$, and $w_k$ denote, respectively, the state, input, and disturbance predicted vectors at prediction step $k$ and at discrete time $t$, $\mathcal{X}_{\mathrm{N}} \subseteq \mathcal{X}$ is the polyhedral robust control invariant terminal set with origin in its interior $\mathbb{0} \in \mathcal{X}_{\mathrm{N}}$. 

n the objective function~\eqref{eq:mpc:theory:MPC_obj}, $\Vert x_{N} \Vert_{Q_{\mathrm{N}}}^\pnorm$ is the weighted terminal penalty, $\Vert x_{k} \Vert_{Q_{\mathrm{x}}}^\pnorm$ is the weighted state vectors, and $\Vert u_{k} \Vert_{Q_{\mathrm{u}}}^\pnorm$ is the weighted input vectors where $Q_{\mathrm{x}} \succeq 0$, $Q_{\mathrm{N}} \succeq 0$, $Q_{\mathrm{u}} \succ 0$, and norm $\pnorm \in \{1,\infty \}$ is assumed.\footnote{For an arbitrary $z \in \mathbb{R}^{n_\text{z}}$ we denote norm $n = 1$ as $\Vert z \Vert_{Q_{\mathrm{N}}}^1 =\sum_{i=1}^{n_\text{z}} |z_i|$ and for $n = \infty$ as $\Vert z \Vert_{Q_{\mathrm{N}}}^{\infty} =\max_{i} |z_i|$.} The terminal penalty $Q_{\mathrm{N}}$ is determined as a solution of Riccati equation of the conventional LQR design problem. The resulting optimization problem~\eqref{eq:mpc:theory:MPC} can be formulated as a linear program (LP) the solution of which yields a sequence of optimal robust control actions $U_\text{N} = [u_0^\top, \dots, u_{N-1}^\top]^\top$.

\begin{assumption}
	\label{asm:MPC}
	Assume that
	\begin{itemize}
		\item matrix pair $(A,B)$ in~\eqref{eq:mpc:theory:sys} is controllable,
		\item sets $\mathcal{X}$, $\mathcal{U}$, $\mathcal{W}$ as in~\eqref{eq:mpc:theory:cons} contain the origin in their respective interiors.
	\end{itemize}
\end{assumption}

By solving~\eqref{eq:mpc:theory:MPC} parametrically, i.e., for all feasible $x(t) \in \mathcal{X}$, we obtain the feedback law $\mu: \Omega \mapsto \mathcal{U}$ and the corresponding value function $J^\star: \Omega \mapsto \mathbb{R}_{\geq 0}$ as continuous PWA functions defined over $\Omega = \cup_{i=1}^M \mathcal{R}_i = \{ x(t) \; | \; \exists u : \eqref{eq:mpc:theory:MPC_cons}-\eqref{eq:mpc:theory:MPC_terminalcon} \ \text{holds}\}$, where  $\mathbb{R}_{\geq 0}$ represents non-negative real numbers, see~\cite{borrelli_bemporad_morari_2017}. Specifically, we define the explicit feedback policy as
\begin{equation}
	\label{eq:mpc:theory:eMPC:mu}
	\mu(x(t)) := F_i x(t) + g_i  \ \quad \text{if} \ \quad x(t) \in \mathcal{R}_i,  
\end{equation}
and the value function as
\begin{equation}
	\label{eq:mpc:theory:eMPC:obj}
	J^\star(x(t)) := H_i x(t) + h_i  \ \quad \text{if} \ \quad x(t) \in \mathcal{R}_i,  
\end{equation}
where $i = 1, \dots, M$.
We say that the polyhedral partition $\Omega$ consists of $M$ regions $\mathcal{R}_i$ with $\interior(\mathcal{R}_i) \cap \interior(\mathcal{R}_j) = \varnothing $, $\forall i \ne j$ holds.\footnote{We denote $\interior(\mathcal{R}_i)$ to be interior of the $i$-th region.}
For the ease of notation, we onward omit the time domain $(t)$, i.e., $x = x(t)$.

\begin{remark}[Quadratic penalty]
	Note, that~\eqref{eq:mpc:theory:MPC} with $\pnorm \in \{1,\infty \}$ leads to an LP, which results in~\eqref{eq:mpc:theory:eMPC:obj}. This specific form of the PWA value function is essential for our further results, see selection of a Lyapunov function in Section~\ref{sec:mpc:app}. 
	We would like to point out that one can still formulate~\eqref{eq:mpc:theory:MPC} as a quadratic problem (QP) with $n=2$, or reformulate~\eqref{eq:mpc:theory:MPC} as an associated QP of the tube MPC~\cite{MS05}.
	Then, however, we need to perform an additional procedure to transform the resulting PWQ value function~\eqref{eq:mpc:theory:MPC_obj} into its corresponding PWA approximation, while preserving its stability properties,  see~\cite[Chapter 13.3.2]{borrelli_bemporad_morari_2017} for instance. 
\end{remark}

\begin{remark}[Stability]
	We recall that by following~\cite{MR00}, one can preserve the closed-loop system stability of~\eqref{eq:mpc:theory:MPC} also without the terminal set constraint~\eqref{eq:mpc:theory:MPC_terminalcon}, e.g., by the proper tuning of the terminal penalty $Q_{\mathrm{N}}$  in the cost function~\eqref{eq:mpc:theory:MPC_obj} and/or by tuning the length of the prediction horizon $N$ to ensure inactive terminal set constraint~\eqref{eq:mpc:theory:MPC_terminalcon}, etc.
\end{remark}

\begin{definition}(Encoded LTI system and control law)
	\label{def:encoded_system}~\\
	Finally, let the uncertain LTI system~\eqref{eq:mpc:theory:sys2} be evaluated above the quantization~\eqref{eq:quant} of control inputs $\hat{u} \in \mathcal{U} \oplus \mathbb{Q}_{\mathrm{u}}$ and system states $\hat{x} \in \mathcal{X} \oplus \mathbb{Q}_{\mathrm{x}}$ 
	such that uncertain LTI system~\eqref{eq:mpc:theory:sys2} is equivalent to the following system in the compact form
	\begin{equation}
		\label{eq:mpc:theory:sys:encrypted}
		\hat{f}(\hat{x}(t),\hat{u}(t),d(t)) = A \hat{x}(t) + B \hat{u}(t) + D d(t),
	\end{equation}
	where $\hat{x}(t) = x(t) + q_\text{x}(t)$ and $\hat{u}(t) = u(t)+q_\text{u}(t)$.
	Then, the corresponding control law of explicit MPC  in~\eqref{eq:mpc:theory:eMPC:mu} is evaluated in its quantized form given by
	\begin{equation}
		\label{eq:mpc:theory:eMPC:mu:quantize}
		\hat{\mu}(\hat{x}(t)) := F_i \hat{x}(t) + g_i  \ \quad \text{if} \ \quad \hat{x}(t) \in \mathcal{R}_i,  
	\end{equation}
	for $\forall i = 1, \dots, M$.
\end{definition}

\begin{lemma}
	\label{lem:stability_encoded}
	Consider explicit control law~\eqref{eq:mpc:theory:eMPC:mu} of robust MPC design problem in~\eqref{eq:mpc:theory:MPC} exists, and Assumption~\ref{asm:MPC} holds. Then the corresponding control law in~\eqref{eq:mpc:theory:eMPC:mu:quantize} ensures closed-loop system stability and recursive feasibility subject to the uncertain LTI system in~\eqref{eq:mpc:theory:sys} for all states $x(t) \in \Omega$ w.r.t. any realization of the quantization error $q_\text{x}(t)$, $q_\text{u}(t)$, and the process disturbance $d(t)$ from $\begin{bmatrix}d^\top(t) & q_\text{x}^\top(t) & q_\text{u}^\top(t) \end{bmatrix}^\top \in \mathcal{W}$.
	
\end{lemma}

\begin{proof}
	Proof consists of two main parts. First, we prove that the system in~\eqref{eq:mpc:theory:sys} is closed-loop stable and recursive feasibility under the control law~\eqref{eq:mpc:theory:eMPC:mu}. Let Assumption~\ref{asm:MPC} hold. It directly follows from \cite[Theorem~15.9]{borrelli_bemporad_morari_2017}, that if there exists a feasible solution of robust MPC design problem in~\eqref{eq:mpc:theory:MPC}, then control law~\eqref{eq:mpc:theory:eMPC:mu} satisfies the closed-system stability of~\eqref{eq:mpc:theory:sys} and recursive feasibility of optimization problem~\eqref{eq:mpc:theory:MPC} for any feasible $x_{0}$. 
	
	Secondly, it remains to prove that these properties are preserved also w.r.t. quantize control law in~\eqref{eq:mpc:theory:eMPC:mu:quantize} constructed above the feasible set of quantize system states $\hat{x}$. 
	As the original uncertain LTI system~\eqref{eq:mpc:theory:sys2} is equivalent to its expanded form given by
	\begin{eqnarray}
		\label{prf:quantize_system}
		f(\cdot) = A x(t) + B u(t) + \underbrace{ A q_{\mathrm{x}}(t) + B q_{\mathrm{u}}(t) + D d(t) }_{W w(t)} ,
	\end{eqnarray}
	for $w(t) = \begin{bmatrix}d^\top(t) & q_\text{x}^\top(t) & q_\text{u}^\top(t) \end{bmatrix}^\top \in \mathbb{R}^{n_\text{d}+n_\text{x}+n_\text{u}}$ and $W = \begin{bmatrix}	D & A & B \end{bmatrix}$, it is directly resulting in the equivalence of~\eqref{prf:quantize_system} to the uncertain LTI system in~\eqref{eq:mpc:theory:sys}. 
	In terms of the equivalent prediction model~\eqref{eq:mpc:theory:MPC_model}, it follows that any feasible solution of~\eqref{eq:mpc:theory:MPC} simultaneously satisfies the recursive feasibility w.r.t. quantize control inputs and system states in~\eqref{prf:quantize_system}, since \eqref{prf:quantize_system} $\Leftrightarrow$~\eqref{eq:mpc:theory:sys2} holds. 
	This leads to~\eqref{prf:quantize_system} being equivalent to~\eqref{eq:mpc:theory:sys:encrypted}. 
	Consequently, according to Definition~\ref{def:encoded_system}, for any $\Vert q_{\mathrm{x}} \Vert_{\infty} \in \mathbb{Q}_{x}$, the associated quantize system state $\hat{x}(t) = x(t) + q_{\mathrm{x}}(t)$ leads to the quantized control law in~\eqref{eq:mpc:theory:eMPC:mu:quantize}, such that the corresponding worst-case quantize error by design satisfies
	
	\begin{subequations}	
		\label{prf:quantize_control_input}
		\begin{align}
			\Vert \hat{\mu}(\hat{x}(t)) - \mu(x(t)) \Vert_{\infty} = \Vert \hat{u}(t) - u(t) \Vert_{\infty} &= \Vert q_{\mathrm{u}} \Vert_{\infty}, \\
			\Vert q_{\mathrm{u}} \Vert_{\infty} &\in \mathbb{Q}_{u}.
		\end{align}
	\end{subequations}	
	Finally, in terms of~\eqref{prf:quantize_system} constructed subject to any impact of the uncertainty $W w(t)$, including the bounded quantize errors for $\forall q_{\mathrm{u}} \in \mathbb{Q}_{\mathrm{u}}$, $\forall q_{\mathrm{x}} \in \mathbb{Q}_{\mathrm{x}}$, and under the equivalence of \eqref{prf:quantize_system} $\Leftrightarrow$ \eqref{eq:mpc:theory:sys}, it directly follows that for any feasible solution $x(t) \in \Omega$ of~\eqref{eq:mpc:theory:MPC}, the corresponding $\hat{\mu}(\hat{x}(t))$ in~\eqref{eq:mpc:theory:eMPC:mu:quantize} leads to the recursive feasible and stabilizing control action $\hat{u}(t)$. \hfill $\Box$
\end{proof}

\begin{remark}[Non-robust case]
	The benefit of the proposed robust MPC design framework is stability guarantees, and, simultaneously, the increased application range subject to the plant-model mismatch and/or other external disturbances. If one needs to reduce the conservativeness by omitting the robust nature of the MPC formulation, then consider $\mathfrak{D}$ in~\eqref{eq:mpc:theory:sys2}, \eqref{eq:mpc:theory:sys:encrypted} to be an empty set.
\end{remark}

\subsection{Parametric Solution with Polynomials}
\label{sec:mpc:app}

In this paper we adopt results from~\cite{uiam1171} based on which we approximate~\eqref{eq:mpc:theory:eMPC:mu} by a new polynomial control law
\begin{equation}
	\label{eq:mpc:app:tmu}
	\widetilde{\mu}(x) = \alpha_1 x + \alpha_2 x^2 + \dots + \alpha_{\pdegree} x^{\pdegree},
\end{equation}
that stabilizes the system~\eqref{eq:mpc:theory:sys} and provides recursive satisfaction of the constraints~\eqref{eq:mpc:theory:cons}. The coefficients to be find are $\boldsymbol{\alpha} = \{\alpha_i, \dots, \alpha_{\pdegree}\}$, $\pdegree$ is the fixed polynomial degree, and $x^i = [x_1^i, x_2^i, \dots, x_{n_{\mathrm{x}}}^i]^\top$ is the element-wise power of the state vector $x \in \mathbb{R}^{n_{\mathrm{x}}}$.

Let us assume that for the closed-loop system $f_{\text{CL}}(x,w) := \widetilde{f}(x,\mu(x),w)$ there exists a PWA Lyapunov function $L : \Omega \mapsto \mathbb{R}_{\geq 0}$ that satisfies $L(f_{\text{CL}}(x,w)) \leq \gamma L(x)$ and $L(\mathbb{0}^{n_\text{x}}) = 0$ for some $\gamma \in [0,1)$, $\forall x \in \Omega$, and $\forall w \in \mathcal{W}$. \footnote{It was shown by many authors~\cite{4026638,MayEtal:aut:00} that by carefully choosing the optimization problem~\eqref{eq:mpc:theory:MPC} then $\mu(x)$ provides closed-loop system stability and the value function $J^\star (x)$ as in~\eqref{eq:mpc:theory:eMPC:obj} can be used as a Lyapunov function for the $f_{\text{CL}}$, i.e., $ L := J^\star (x)$.}

For the system~\eqref{eq:mpc:theory:sys}, Lyapunov function $L(\cdot)$, and a fixed $\gamma$ we construct stabilizing tube~\cite[Chapter 10.4]{christophersen2007optimal}:
\begin{equation}
	\label{eq:mpc:app:tube}
	\begin{split}
		\mathcal{T}(L,\gamma) \! := \!
		\begin{Bmatrix}
			\left[x^\top, u^\top \right]^\top \, \!\! \vert & \!\!\!\!\!\!\!\!\!\! \, u \in \mathcal{U}, x \in \Omega, \widetilde{f}(x,u,w) \in \Omega, \\
			& \!\!\!\!\!\! L(\widetilde{f}(x,u,w)) \leq \gamma L(x), \forall w \in \mathcal{W}
		\end{Bmatrix} \! .
	\end{split}
\end{equation}

As $\text{dom}(L) = \Omega = \cup_{i=1}^M \mathcal{R}_i$, we have that $\mathcal{T}(L,\gamma) := \cup_{i=1}^{M} \mathcal{T}_i$ with $\mathcal{T}_i := \cup_{j=1}^{M} \mathcal{T}_{i,j}$ for all $i = 1, \dots, M$. Respectively, the entire stability tube is defined as union of polytopes $\mathcal{T}(L,\gamma) := \cup_{i=1}^{M} \cup_{j=1}^{M} \mathcal{T}_{i,j}$ where each $\mathcal{T}_{i,j} \in \mathbb{R}^{n_x+n_u}$ represents a feasible transition from $\mathcal{R}_i$ to $\mathcal{R}_j$ for which the value of $L(\cdot)$ decreases.

Let us define each region $\mathcal{R}_i$ of $\Omega$ in its vertex representation
\begin{align}
	\label{eq:mpc:app:region_vertex}
	\mathcal{R}_i &= \begin{Bmatrix} x \, | \, x = \sum_{j=1}^{|\nu_i|} \lambda_j |\nu_i|_j, \lambda \in \Lambda_i \end{Bmatrix},\\
	\Lambda_i &= \begin{Bmatrix} \lambda \, | \,  0 \leq \lambda_j \leq 1, \sum_{j=1}^{|\nu_i|}\lambda_j=1 \end{Bmatrix},
\end{align}
where $\nu_i$ are vertices of $\mathcal{R}_i$, $|\nu_i|_j$ is the $j$-th vertex of $\mathcal{R}_i$, and $|\cdot|$ denotes cardinality.
Assume that for a given $V(\cdot)$ and a fixed $\gamma \in [0,1)$ there exists a full-dimensional stability tube $\mathcal{T}(L,\gamma)$ that is composed of $i = 1,\dots, M$ polytopes\footnote{It is required that {$\mathcal{T}_i$ is a polytope} or there exits an inner polytopic approximation that covers the entire state space of $\mathcal{R}_i$, i.e., $\text{proj}_x\mathcal{T}_i = \mathcal{R}_i$.} 
\begin{equation}
	\label{eq:mpc:app:subtube}
	\mathcal{T}_i := \begin{Bmatrix}\left[x^\top, u^\top\right]^\top \, | \, \left[T^{\mathrm{x}}_{i} T^{\mathrm{u}}_{i}\right] \left[x^\top, u^\top\right]^\top \leq T^0_i \end{Bmatrix},
\end{equation}
with the connected union $\cup_{i=1}^{M} \mathcal{T}_i$.
By plugging~\eqref{eq:mpc:app:tmu} into~\eqref{eq:mpc:app:subtube} and substituting the state vector $x$ by vertices of $\mathcal{R}_i$ per~\eqref{eq:mpc:app:region_vertex} we obtain a set of polynomials 
\begin{equation}
	\label{eq:mpc:app:poly}
	\rho_i(\boldsymbol{\alpha},\lambda) := T^0_i - T^x_i \lambda - T^u_i \widetilde{\mu}(\lambda), \, \quad \lambda \in \Lambda_i,
\end{equation}
for all $i = 1, \dots, M$. 
These polynomials can be further homogenized by multiplying single monomials with $\left(\sum_{j=1}^{|\nu_i|}\lambda_j\right)$ until all have the same degree. The Pólya's polynomials can be then defined as follows
\begin{equation}
	\label{eq:mpc:app:homogen}
	\rho_i^\varphi(\boldsymbol{\alpha},\lambda) = \rho_i(\boldsymbol{\alpha},\lambda) \cdot \left(\sum_{j=1}^{|\nu_i|}\lambda_j\right)^\varphi,
\end{equation}
where $\varphi$ is a Pólya degree, see~\cite{MOK2008524}.
%

Finally, as shown in~\cite{uiam1171}, if a stabilizing Lyapunov function $L(\cdot)$ and a stability tube $\mathcal{T}(L,\gamma)$ satisfying all the aforementioned assumptions can be constructed, then the coefficients $c_i^\varphi$ of the Pólya polynomials~\eqref{eq:mpc:app:homogen} can be defined symbolically. Subsequently, the following linear program
\begin{subequations}
	\label{eq:mpc:app:opt}
	\begin{align}
		\min_{\boldsymbol{\alpha}} \;& z(\cdot) \label{eq:mpc:app:opt_obj}\\
		\text{s.t.} \ 	&c_i^\varphi \geq 0 \, , \, i = 1,\dots, M 	\label{eq:mpc:app:opt:cons}  
	\end{align}
\end{subequations}
can be solved to obtain coefficients $\boldsymbol{\alpha}$ of the approximated control law $\widetilde{\mu}(x)$ as in~\eqref{eq:mpc:app:tmu} that stabilizes the system~\eqref{eq:mpc:theory:sys}, while providing recursive feasibility of the original constraints in~\eqref{eq:mpc:theory:cons}. 
Note, there are multiple alternatives how to define the objective function~\eqref{eq:mpc:app:opt_obj}. The most straightforward one is to choose $z(\cdot) = 0$, i.e., state~\eqref{eq:mpc:app:opt} as a pure feasibility problem. Another option is to minimize the point-wise distance between the original and the approximated control policy over selected points from $\Omega$. One suitable candidate is to choose vertices of each region $\mathcal{R}_i$, i.e., $z(\cdot) = \sum_{i=1}^M \sum_{j=1}^{|\nu_i|} ||\mu(\nu_{i,j}) - \widetilde{\mu}(\nu_{i,j})||^\pnorm$ with $\pnorm \in {\{1,\inf\}}$.

\begin{corollary}
	Consider the approximated control law~\eqref{eq:mpc:app:tmu} of robust MPC design problem in~\eqref{eq:mpc:theory:MPC} is given to satisfy~\eqref{eq:mpc:app:opt}, and let Assumption~\ref{asm:MPC} hold. If approximated control law $\widetilde{\mu}(x)$ in~\eqref{eq:mpc:app:tmu} ensures the closed-loop system stability and recursive feasibility of~\eqref{eq:mpc:theory:MPC} for any $x(t) \in \Omega$, then also $\widetilde{\mu}(\hat{x})$ according to Definition~\ref{def:encoded_system} preserves these properties w.r.t. any realization of the bounded quantization error $q_\text{x}(t)$, $q_\text{u}(t)$, and the process disturbance $d(t)$ from $\begin{bmatrix}d^\top(t) & q_\text{x}^\top(t) & q_\text{u}^\top(t) \end{bmatrix}^\top \in \mathcal{W}$.
\end{corollary}

\begin{proof}
	The proof has two parts. First, we need to prove that the approximated control law $\widetilde{\mu}(x)$ in~\eqref{eq:mpc:app:tmu} satisfies the closed-loop system stability and the recursive feasibility of~\eqref{eq:mpc:theory:MPC} for any feasible $x \in \Omega$. We recall Lemma~3.5. in~\cite{uiam1171} implying that for any feasible solution $\boldsymbol{\alpha}$ of~\eqref{eq:mpc:app:opt} the resulting control law $\widetilde{\mu}(x)$ in~\eqref{eq:mpc:app:tmu} preserves by design the closed-loop stability of system~\eqref{eq:mpc:theory:sys2} and the recursive feasibility subject to the constrained system states and control inputs in~\eqref{eq:mpc:theory:MPC}. 
	
	Secondly, it remains to prove that also the approximated control law $\widetilde{\mu}(\hat{x})$ in~\eqref{eq:mpc:app:tmu} evaluated for quantize system state $\hat{x}(t) = x(t) + q_{\mathrm{x}}(t)$ preserves the closed-loop system stability and recursive feasibility for any $q_{\mathrm{x}} \in \mathbb{Q}_{x}$, $q_{\mathrm{u}} \in \mathbb{Q}_{u}$, and $d \in \mathfrak{D}$ in system~\eqref{eq:mpc:theory:sys}.
	By substituting any bounded quantization errors $q_{\mathrm{u}} \in \mathbb{Q}_{u}$, $q_{\mathrm{x}} \in \mathbb{Q}_{x}$ into~\eqref{prf:quantize_system}, it follows from Lemma~\ref{lem:stability_encoded} and the from equivalence of~\eqref{prf:quantize_system} and~\eqref{eq:mpc:theory:sys}, that
	$\widetilde{\mu}$ by design ensures closed-loop system stability and recursive feasibility of $\widetilde{\mu}(\hat{x})$ in~\eqref{eq:mpc:app:tmu} for any feasible $\hat{x} \in \Omega$.  \hfill $\Box$
\end{proof}

\section{Experimental Validation of Encrypted MPC}
\label{sec:CaseStudy}

This section presents a \textit{step-by-step} implementation of secured control with encrypted control law in the form of~\eqref{eq:mpc:app:tmu} in a closed-loop 
control scenario. We explore the parameter settings of the BFV implementation and test it on a laboratory device, where the applicability of the polynomial 
approximation is clearly visible as well as the effects of the BFV settings on the closed-loop system. The \texttt{Python} language was used to simulate and 
evaluate results from encrypted closed-loop process control. All experiments were performed on a $64$-bit operating system $16$-core CPU with 
$\SI{3.4}{\giga\hertz}$ processor and $\SI{128}{\giga\byte}$ of RAM. All computations were performed in a single-threaded mode. 

\subsection{Laboratory Device Flexy $2.0$}
\label{sec:device}

We demonstrate the presented control algorithm on single-input single-output laboratory device Flexy~$2.0$ designed as a benchmark laboratory system for control applications{, with detailed parameters given in~\cite{kaluz:2019:flexy}}. Manipulated variables are 
represented by fan speed defined in percentage. The fan produces an airflow bending the flex resistor, thus changing the electric resistance which is the output 
variable measured in percentage. The appearance of Flexy~$2.0$ is depicted in Figure~\ref{flexy}.

\begin{figure}[ht]
	\centering
	\subfigure[Laboratory device]{{\includegraphics[width=0.45\linewidth, trim={3cm 2cm 2.5cm 3cm}, clip]{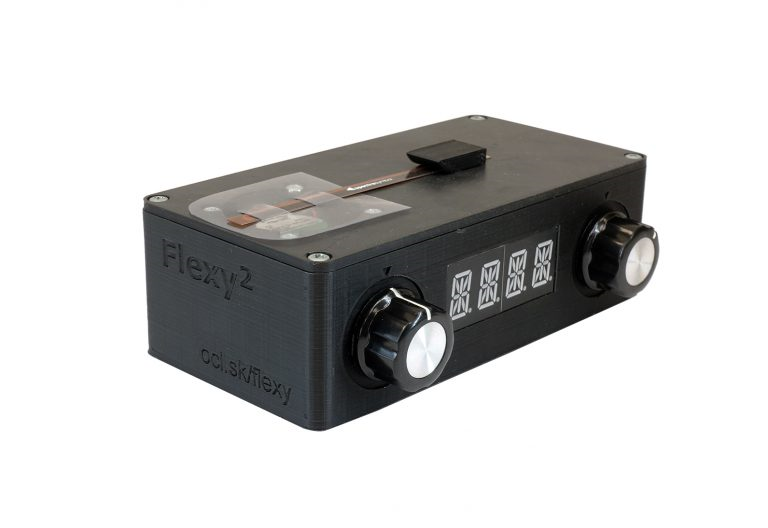}  }}
	\subfigure[{Side-profile sketch}]{{\includegraphics[width=0.45\linewidth]{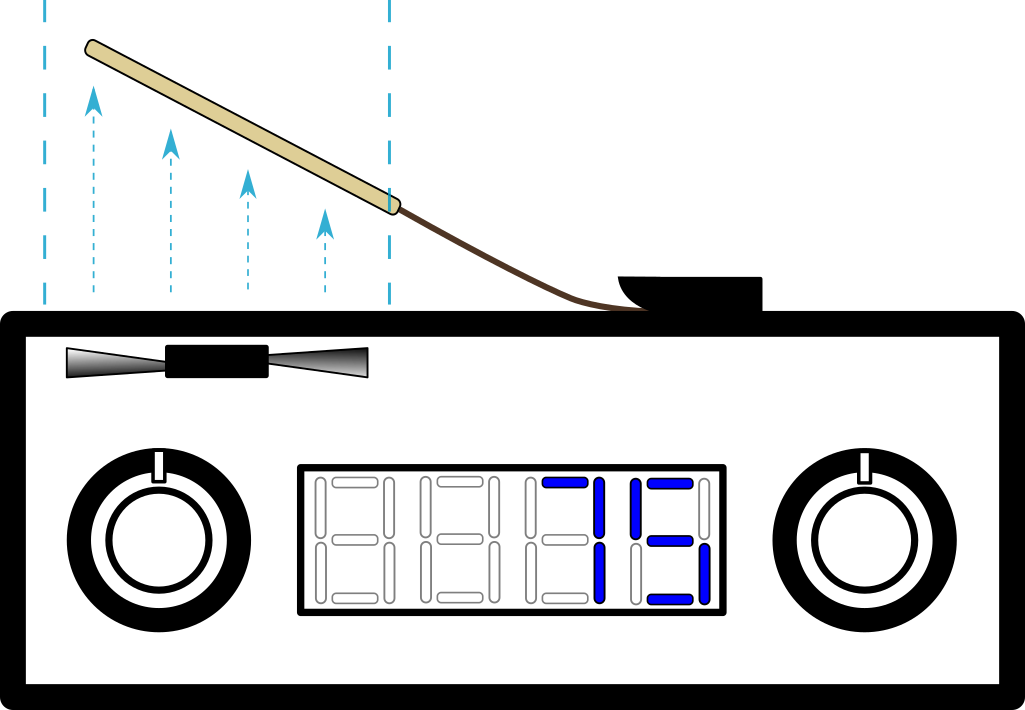}  }}
	\caption{{Laboratory device Flexy~$2.0$ with a flex-bend as controlled variables and fan-speed as manipulated variable.}}
	\label{flexy}
\end{figure}

We adopt the linear model of Flexy $2.0$ from~\cite{Flexy2} in the form of a continuous transfer function defined as
\begin{equation}
	G(s) = \frac{0.89}{0.66s + 1}
\end{equation}  
which is further transformed into a discrete time state space model

\begin{equation}
	x(t+T_\text{s}) = 0.966 \, x(t) + 0.101 \, u(t),
	\label{eq:process}
\end{equation}
with sampling period $ T_\text{s} = \SI{10}{\milli\second}$. 
By considering both quantization and process disturbance errors, i.e., $w(t) = \begin{bmatrix}d^\top(t) & q_\text{x}^\top(t) & q_\text{u}^\top(t) \end{bmatrix}^\top$, the final uncertain linear time-invariant model~\eqref{eq:mpc:theory:sys} is represented by
\begin{equation}
	\begin{array}{rl}
		x(t+T_\text{s}) = &0.966 \, x(t) + 0.101 \, u(t) + \\
		&\underbrace{\begin{bmatrix} 1.000 & 0.966 & 0.101 \end{bmatrix}}_W\, \ w(t),
		\label{eq:process}
	\end{array}
\end{equation}
and is utilized for the MPC design in~\eqref{eq:mpc:theory:MPC_model}.

Flexy device represents a laboratory-scale benchmark for systems that mimic the challenges of controlling industrial devices such as: 
\begin{itemize}
	\item satisfying physical constraints,
	\item handling the systems with fast dynamics,
	\item signal filtering and processing,
	\item limitation of the device's microcontroller to handle measurements and apply control actions,
\end{itemize}
while respecting the sampling period at the level of milliseconds. 


\subsection{Encrypted Control Setup}
\label{subsec:control_setup}

This section describes the control setup suitable for the use of the BFV scheme in secured process control. 
The encrypted MPC approach is implemented according to the scheme in Figure~\ref{fig:encrypted_control_scheme}. 

\begin{figure}[h!]
	\centering
	\includegraphics[width=0.9\linewidth]{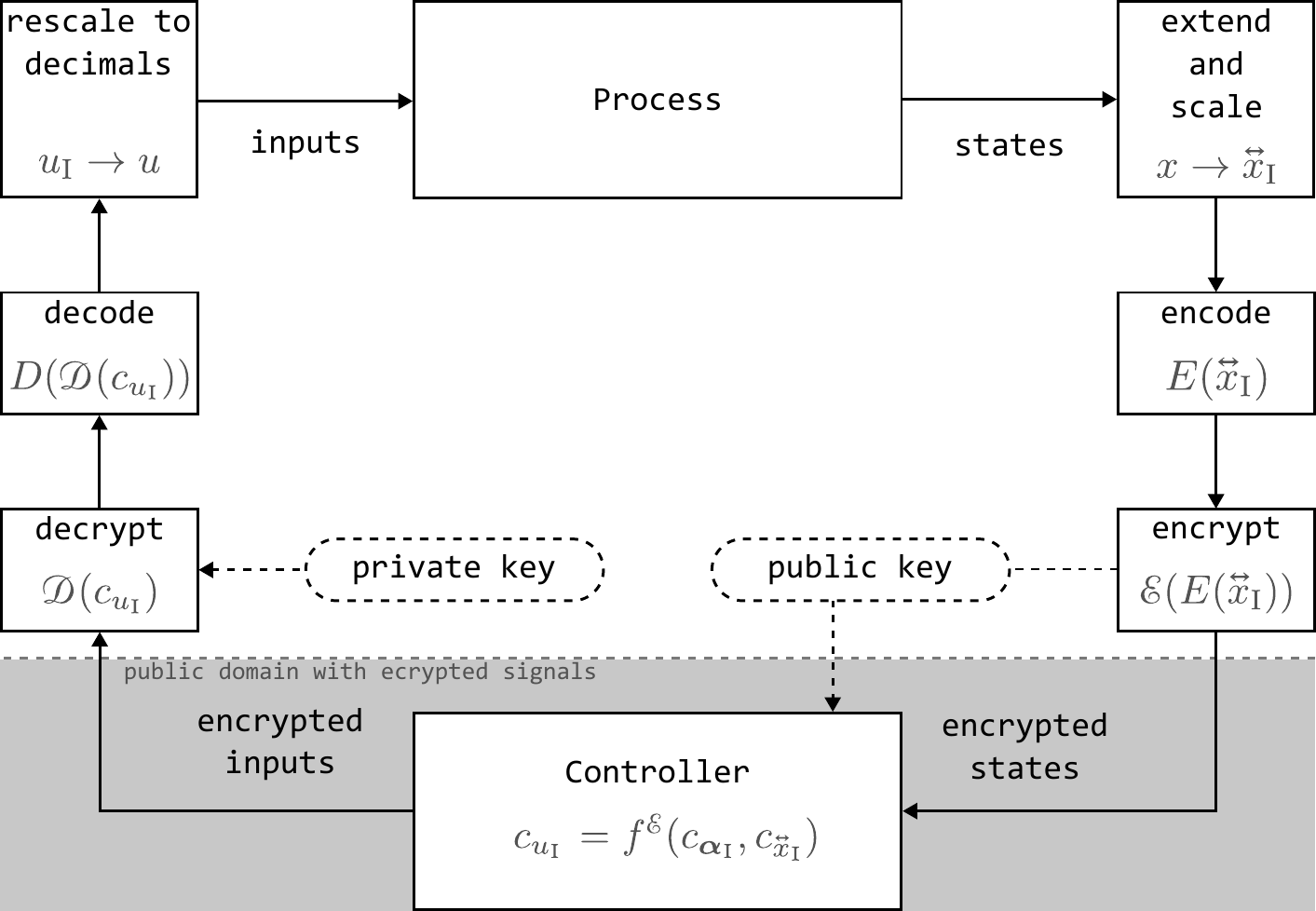}
	\caption{Block diagram of the encrypted polynomial control setup.}
	\label{fig:encrypted_control_scheme}
\end{figure}

We consider \texttt{TenSEAL} library~\cite{tenseal2021} which is a \texttt{Python} implementation of \texttt{Microsoft SEAL}~\cite{sealcrypto}. For use of the BFV scheme, the following parameters have to be set:
\begin{itemize}
	\item \textit{Polynomial modulus degree ($\mathcal{N}$)} defining the maximal degree of each polynomial in polynomial space $R$. Parameter $\mathcal{N}$ has to be positive and a power of $2$, typically $\mathcal{N} \in \{ 2\,048, 4\,096, \dots \} $. The security level increases with growing $\mathcal{N}$. However, higher $\mathcal{N}$ means higher computational complexity, thus the evaluation of mathematical operations over ciphertexts is slower. In \texttt{TenSEAL} is degree $\mathcal{N}$ a mandatory parameter, thus it has to be defined.
	\item \textit{Plaintext modulus ($\tau$)} used to handle the coefficient size of plaintexts. This parameter has to be a prime and basically represents the largest number that can be encoded and encrypted. The value of $\tau$ is selected such that $\mod(2\mathcal{N},\tau)=1$. Plaintext modulus $\tau$ is a mandatory parameter. 
\end{itemize} 
Two parameters $\mathcal{N}$ and $\tau$ are elements of a structure called \textit{context}, and are user-defined settings.
Finally, to perform homomorphic multiplications we generate relinearization(evaluation) keys to handle the size of resulting ciphertexts.

Since the majority of control applications operate with floating point values, we adjust states and controller coefficients into integers with relation~\eqref{eq:quant}, hence
\begin{subequations}
	\begin{align}
		\ui{x}{I} &= \nint{x \cdot 10^{\ui{\theta}{x}}}, \label{eq:state:int}\\
		\ui{\boldsymbol{\alpha}}{I} &= \nint{\boldsymbol{\alpha} \cdot 10^{\theta_{\alpha}}}, \label{eq:ctrl:int}
	\end{align}
\end{subequations}
where in $\ui{x}{I}$ is integer value of state $x$ with precision degree $\ui{\theta}{x}$ and $\ui{\boldsymbol{\alpha}}{I}$ is polynomial controller with integer coefficients with precision $\theta_{\alpha}$.

Next, we present the implementation of encrypted polynomial closed-loop control for single-input single-output systems:
\begin{enumerate}
	\item \textit{Modify states}: first, we adjust the measured state $x$ to vector form $\mathbf{x}$ with respect to the order of the polynomial controller, hence
	\begin{equation}
		\mathbf{x} = \begin{bmatrix} x, & x^2, & \dots, & x^\kappa \end{bmatrix},
		\label{eq:states:extend}
	\end{equation}
	where $\kappa$ is the order of the polynomial controller.
	\item \textit{Adjust to integers}: next, we convert $\mathbf{x}$ to integer vector $\ui{\mathbf{x}}{I}$ with precision degree $\ui{\theta}{x}$ by using~\eqref{eq:state:int} element-wise, hence
	\begin{equation}
		\ui{\mathbf{x}}{I} = \nint{\mathbf{x} \cdot 10^{\ui{\theta}{x}}}.
		\label{eq:states:ajust}
	\end{equation}
	We encode and encrypt $\ui{\mathbf{x}}{I}$ send ciphertext \mbox{$c_{\ui{\mathbf{x}}{I}} = \mathcal{E}(E(\ui{\mathbf{x}}{I}))$} to the control unit. 
	\item \textit{Evaluate control law}: the encrypted controller is represented by ciphertext $c_{\ui{\boldsymbol{{\alpha}}}{I}} = \mathcal{E}(E(\ui{\boldsymbol{{\alpha}}}{I}))$ obtained by encrypting polynomial controller in integer form $\ui{\boldsymbol{{\alpha}}}{I}$~\eqref{eq:ctrl:int} with precision degree $\theta_{\alpha}$. We evaluate the control law~\eqref{eq:mpc:app:tmu} over ciphertexts $c_{\ui{\mathbf{x}}{I}}$ and $c_{\ui{\boldsymbol{{\alpha}}}{I}}$ and return ciphertext of the encrypted control input $c_{\ui{u}{I}}$ computed as
	\begin{equation}
		c_{\ui{u}{I}} = f^{\mathcal{E}}(c_{\ui{\boldsymbol{{\alpha}}}{I}}, c_{\ui{\mathbf{x}}{I}}) = 
		\sum_{i=1}^{\kappa} c_{\ui{\boldsymbol{{\alpha}}}{I}, (i)} \cdot c_{\ui{\mathbf{x}}{I}, (i)}				
		\label{eq:sec:ctrl}
	\end{equation}
	which basically represents the scalar product of two encrypted vectors represented by ciphertexts $c_{\ui{\boldsymbol{{\alpha}}}{I}}$ and $c_{\ui{\mathbf{x}}{I}}$.
	\item \textit{Rescale control input}: encrypted control input is decrypted and rescaled back by relation
	\begin{equation}
		u = \frac{D(\mathcal{D}(c_{\ui{u}{I}}))}{10^{(\ui{\theta}{x}+\theta_{\alpha})}},
	\end{equation}
	where $u$ is the control input ready to be applied to the process. The control algorithm~\ref{alg:poly:ctrl} includes all of the necessary steps for encrypted polynomial control considering BFV framework. 	
	
	\begin{algorithm}
		\caption{Encrypted Polynomial Control}
		\label{alg:poly:ctrl}
		\begin{algorithmic}[1]
			\State Measure state $x_k$
			\State Extend to vector $\mathbf{x}_k$
			\State Scale to integers $\mathbf{x}_{\text{I}, k} = \nint{\mathbf{x}_k\cdot 10^{\ui{\theta}{x}}}$
			\State Encode and encrypt $c_{\ui{\mathbf{x}}{I}, k} = \mathcalE(E(\mathbf{x}_{\text{I}, k}))$
			\State Evaluate control law $c_{\ui{u}{I}, k} = f^{\mathcalE}(c_{\ui{\boldsymbol{\alpha}}{I}}, c_{\ui{\mathbf{x}}{I}, k})$
			\State Decrypt and decode $\mathbf{x}_{\text{I}, k} = D(\mathcalD(c_{\ui{u}{I}, k}))$
			\State Rescale $u_k = \frac{\mathbf{x}_{\text{I}, k}}{10^{(\ui{\theta}{x} + \theta_{\alpha})}}$
			\State Apply $u_k$
		\end{algorithmic}
	\end{algorithm}
	
	Note that both polynomial controllers and states are represented by a single ciphertext and controller $\boldsymbol{\alpha}$ is adjusted to integers and encrypted only once. The control law evaluated over ciphertexts (Algorithm~$1$, line~$7$) respects the relation~\eqref{eq:mpc:app:tmu}.
	
\end{enumerate}

The presented control setup will ensure the encrypted polynomial control with the use of the BFV cryptographic scheme.

\subsection{Results and Discussion}
\label{sec:results}

We construct an MPC controller for system~\eqref{eq:process}, where states are subjected to $|x| \leq 4$ and control inputs to $|u| \leq 1$.
By solving~\eqref{eq:mpc:theory:MPC} with $\pnorm = 1$, $N = 10$, $Q_{\mathrm{u}} = 1$, $Q_{\mathrm{x}} = 10$, and $\mathcal{W}$ was such that $-0.05 \leq {\begin{bmatrix} 1.000 & 0.966 & 0.101 \end{bmatrix}} w(t) \leq 0.05$, we obtained the explicit feedback law as in~\eqref{eq:mpc:theory:eMPC:mu} defined over $M = 14$ regions {constructed in $0.4856$ seconds. The memory footprint of this controller was $0.84$\,kB.} By choosing $J^\star(\cdot) = L(\cdot)$ and selecting $\gamma = 0.99$ we synthesized the stabilizing tube~\eqref{eq:mpc:app:tube}. The LP problem in~\eqref{eq:mpc:app:opt} was constructed in \texttt{YALMIP}~\cite{yalmip:paper} with $\pdegree = 3$, $\varphi = 1$, and  $z(\cdot) = \sum_{i=1}^M \sum_{j=1}^{|\omega_i|} ||\mu(\omega_{i,j}) - \widetilde{\mu}(\omega_{i,j})||^1$, where $\omega_{i}$ is a grid of $10$ equidistantly split points across the region $\mathcal{R}_i$, i.e., $|\omega_i| = 10$. The solution of~\eqref{eq:mpc:app:opt}, {formulated in \texttt{YALMIP} within $7.29$ seconds and solved via \texttt{GUROBI} in $1.11$ seconds}, was the stabilizing polynomial feedback law as in~\eqref{eq:mpc:app:tmu} with 
\begin{equation}
	\boldsymbol{\alpha} = \begin{bmatrix} 0.0000, -2.3110, 0.0098, 0.00078 \end{bmatrix},
	\label{eq:controller}
\end{equation}
with the memory footprint of $0.2$ kB.
To determine the suboptimality of $\widetilde{\mu}(\cdot)$ w.r.t. $\mu(\cdot)$, we have split the domain $\Omega$ into $100$ equidistantly divided points and used them as initial conditions for closed-loop profiles. The suboptimality was computed as $(J^\star - \widetilde{J})/J^\star \, 100\,\% = 56\,\%$ where $J^\star$ and $\widetilde{J}$ are the overall values of~\eqref{eq:mpc:theory:MPC_obj} for the feedback law $\mu(\cdot)$ and $\widetilde{\mu}(\cdot)$, respectively, within all performed simulations. 
The graphical comparison of these controllers is presented in Figure~\ref{fig:approxeMPC}.

\begin{figure}[ht]
	\centering
	\setlength\figureheight{5cm}
	\setlength\figurewidth{0.7\linewidth}
	\input{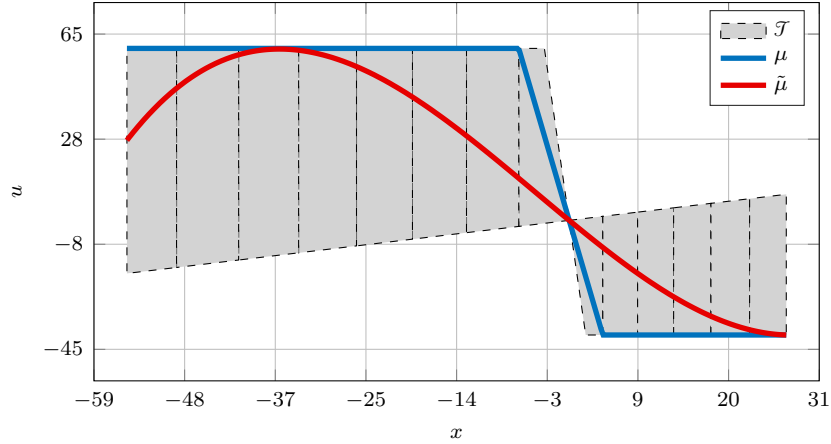}\label{fig:control:ctrl}
	\caption{Approximated explicit MPC policy. Here, the optimal feedback law $\mu(\cdot)$ is represented by the blue line, the approximated polynomial control law $\widetilde{\mu}$ be red-dashed line, and the stabilizing tube $\mathcal{T}$ is defined by the gray area.}
	\label{fig:approxeMPC}
\end{figure}

We consider a control law~\eqref{eq:mpc:app:tmu} with polynomial controller~\eqref{eq:controller} in the task of disturbance rejection, while respecting the steady state values $\ui{x}{s} = \SI{73}{\percent}$ and $\ui{u}{s} = \SI{50}{\percent}$. The experiment of encrypted polynomial control included two various configurations of the BFV framework. We fix the polynomial modulus degree $\mathcal{N} = 4\,096$ and generate two different values of plaintext modulus $\tau$ for given $\mathcal{N}$ and precision degrees $\ui{\theta}{x}$ and $\theta_{\alpha}$ with Algorithm~\ref{alg:plain:mod:gener}
\begin{algorithm}
	\caption{Plaintext modulus generation}
	\label{alg:plain:mod:gener}
	\begin{algorithmic}[1]
		\State $\tau \gets 10^{(\ui{\theta}{x}+\theta_{\alpha}+1)}$
		\While{$(2\mathcal{N} \bmod \tau) \neq 1$}
		\State $\tau \gets \texttt{nextPrime}(\tau)$
		\EndWhile
	\end{algorithmic}
\end{algorithm}
where the initial value of $\tau$ consists of precision for state measurement $\ui{\theta}{x}$, precision for polynomial controller $\theta_{\alpha}$ and some additional space for computations represented by adding $1$. Function $\texttt{nextPrime}(v)$ generates prime larger than $v$. For each experiment, we measure the maximal $\ui{t}{max}$ and average $\ui{t}{avg}$ computational time of evaluating control law~\eqref{eq:sec:ctrl} over ciphertexts. The analysis was realized for time span of $\ui{t}{exp} = \SI{30}{\second}$ during which we introduced a control input disturbances $d_1 = \SI{0}{\percent}$ and $d_2 = \SI{100}{\percent}$ at times $t_{\text{d}_1} = \SI{10}{\second}$ and $t_{\text{d}_2} = \SI{20}{\second}$, respectively. Both disturbances lasted $\SI{0.5}{\second}$. The experimental setups $\mathcal{S}_1$ and $\mathcal{S}_2$ along with their computational performance results are listed in Table~\ref{tab:setup:control}. The graphical results of both experiments are presented in Figure~\ref{fig:control:perf}.

\begin{table}[h!]
	\caption{Experimental encrypted control setup with corresponding computational times.}
	\label{tab:setup:control}
	\centering
	\begin{tabular}{c|c|c|c|c|c|c}
		\toprule
		\thead{Setup} & \thead{$\mathcal{N}$} & \thead{$\tau$} & \thead{$\ui{\theta}{x}$} & \thead{$\theta_{\alpha}$} & \thead{$\ui{t}{avg}$ \\ $[\SI{}{\milli\second}] $} & \thead{$\ui{t}{max}$ \\ $[\SI{}{\milli\second}] $} \\
		\midrule 
		$S_1$ & $4\,096$ & $1\,032\,193$   & $1$ & $4$ & $\SI{3.4}{}$ & $\SI{6.3}{}$ \\
		$S_2$ & $4\,096$ & $100\,016\,129$ & $3$ & $4$ & $\SI{3.3}{}$ & $\SI{5.5}{}$ \\ 
		\bottomrule	
	\end{tabular} 	
\end{table}

\begin{figure}[ht]
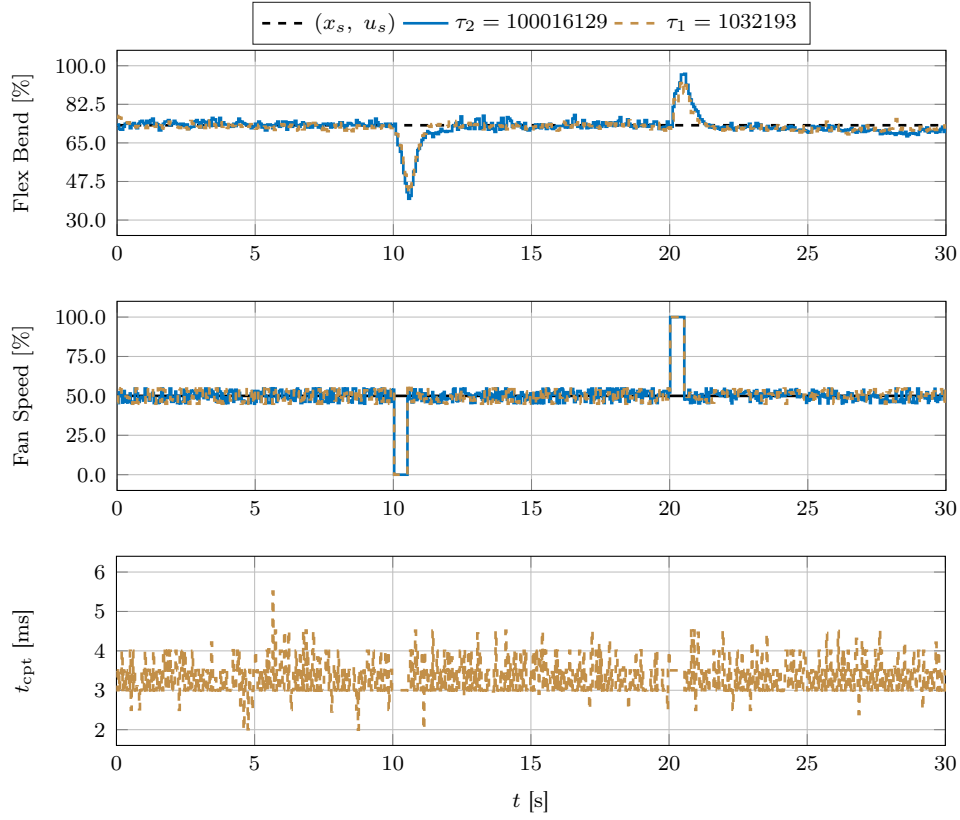

	\centering
	\setlength\figureheight{2.5cm}
	\setlength\figurewidth{.8\linewidth}
	\subfigure{\input{images/x.tikz}\label{fig:control:x}}
	\subfigure{\input{images/u.tikz}\label{fig:control:u}}
	\subfigure{\input{images/cpt.tikz}\label{fig:control:cpt}}
	\caption{Control performance of encrypted polynomial control. The upper graph represents the state profile acquired by control inputs depicted in the middle graph. The lower graph depicts the computational times of evaluating control law over ciphertexts.}
	\label{fig:control:perf}	
\end{figure}

Reading out from the results depicted in Figure~\ref{fig:control:perf} it is obvious that the presented control algorithm worked well with both setups of the BFV scheme. For $\tau_1=1\,032\,193$, the maximal computational time was $\ui{t}{max} = \SI{6.3}{\milli\second}$ and the average time $\ui{t}{avg} = \SI{3.4}{\milli\second}$. Both times remained within the interval of sampling period $\ui{T}{s} = \SI{10}{\milli\second}$. For the setup with $\tau_2=100\,016\,129$ were the computational times at the same level as for the first setup ($\ui{t}{max} = \SI{5.5}{\milli\second}$, $\ui{t}{avg} = \SI{3.3}{\milli\second}$), since $\mathcal{N}$ is fixed, thus the computational complexity remains unchanged.
Finally, we have also computed average quantization errors of states $q_\text{x}$ and control actions $q_\text{u}$. For the setup with $\tau_1=1\,032\,193$ we have detected $q_\text{x} = 170.0\cdot10^{-4}$ and $q_\text{u} = 41\cdot10^{-4}$, while for the setup $\tau_2=100\,016\,129$ we obtained $q_\text{x} = 1.7\cdot10^{-4}$ and $q_\text{u} = 0.6\cdot10^{-4}$.

\section{Conclusions}\label{sec:Conclusion}

Addressing data and controller privacy with homomorphic encryption strategies and considering explicit model predictive control and its versatile deployment, the proposed novel control method eliminates the complexity of online optimization or the necessity to solve the point location problems in encryption frameworks. By merging the fully homomorphic encryption frameworks with the polynomial approximation of the explicit model predictive controller, we achieved cryptographically secured close-to-optimal constrained process control, which has not been done so far. This manuscript also presents an experimental implementation and validation in a laboratory-scaled setup that confirms the viability of this encrypted polynomial control, guaranteeing stability, recursive feasibility, and data security. This research not only significantly contributes to the academic domain but also provides industries with a path to data safety in cloud-based control applications.

\section*{Acknowledgments}
\label{sec:acknowledgments}

The authors gratefully acknowledge the contribution of the Scientific Grant Agency of the Slovak Republic under the grants 1/0339/26, 1/0239/24, the Slovak Research and Development Agency under the project APVV-24-0007, and the Research and Innovation Authority (VAIA) under the grant no. 09I01-03-V04-00024 (Slovak Research Excellence in Advanced Control for Smart Industries). M.~Klaučo is also supported by the European Union project ROBOPROX (Reg. No. CZ.02.01.01/00/22\_008/0004590).

\bibliographystyle{elsarticle-num}
\bibliography{references}

@ARTICLE{teranishi:2024:tcns,
  author={Teranishi, Kaoru and Sadamoto, Tomonori and Kogiso, Kiminao},
  journal={IEEE Transactions on Control of Network Systems}, 
  title={Input–Output History Feedback Controller for Encrypted Control With Leveled Fully Homomorphic Encryption}, 
  year={2024},
  volume={11},
  number={1},
  pages={271-283},
  doi={10.1109/TCNS.2023.3280460}}

@article{Darup2021,
  author = {Darup, M. S. and Alexandru, A. B. and Quevedo, D. E. and Pappas, G. J.},
  journal = {IEEE Control Systems Magazine}, 
  title = {Encrypted Control for Networked Systems: An Illustrative Introduction and Current Challenges}, 
  year = {2021},
  volume = {41},
  number = {3},
  pages = {58-78},
  doi = {10.1109/MCS.2021.3062956}
}

@ARTICLE{7165580,
  author={Singh, J. and Pasquier, T. and Bacon, J. and Ko, H. and Eyers, D.},
  journal={IEEE Internet of Things Journal}, 
  title={Twenty Security Considerations for Cloud-Supported Internet of Things}, 
  year={2016},
  volume={3},
  number={3},
  pages={269-284},
  doi={10.1109/JIOT.2015.2460333}}

@incollection{KNAPP201541,
  title = {Chapter 3 - Industrial Cyber Security History and Trends},
  author = {E. D. Knapp and J. T. Langill},
  booktitle = {Industrial Network Security (Second Edition)},
  publisher = {Syngress},
  edition = {Second Edition},
  address = {Boston},
  pages = {41-57},
  year = {2015},
  isbn = {978-0-12-420114-9},
  doi = {https://doi.org/10.1016/B978-0-12-420114-9.00003-4}
}

@misc{BFV2012,
  author = {Fan, J. and Vercauteren, F.},
  title = {Somewhat Practical Fully Homomorphic Encryption},
  howpublished = {Cryptology ePrint Archive, Report 2012/144},
  year = {2012}
}

@ARTICLE{shen:2021:smcs,  
  author={Shen, X. and Li, X.},  
  journal={IEEE Transactions on Systems, Man, and Cybernetics: Systems},   
  title={Data-Driven Output-Feedback {LQ} Secure Control for Unknown Cyber-Physical Systems Against Sparse Actuator Attacks},   
  year={2021},  
  volume={51},  
  number={9},  
  pages={5708-5720},  
  doi={10.1109/TSMC.2019.2957146}
}

@INPROCEEDINGS{stobbe:2022:cdc,
  author={Stobbe, Pieter and Keijzer, Twan and Ferrari, Riccardo M.G.},
  booktitle={2022 IEEE 61st Conference on Decision and Control (CDC)}, 
  title={A Fully Homomorphic Encryption Scheme for Real-Time Safe Control}, 
  year={2022},
  pages={2911-2916},
  doi={10.1109/CDC51059.2022.9993055}
}

@article{Darup2020,
  author = {Darup, M. S.},
  title = {Encrypted polynomial control based on tailored two-party computation},
  journal = {International Journal of Robust and Nonlinear Control},
  volume = {30},
  number = {11},
  pages = {4168-4187},
  doi = {https://doi.org/10.1002/rnc.5003},
  year = {2020}
}

@ARTICLE{Feng2023,
  author={Feng, Zhaowen and Cao, Guoyan and Grigoriadis, Karolos M. and Pan, Quan},
  journal={IEEE Transactions on Industrial Informatics}, 
  title={Secure {MPC}-Based Path Following for UAS in Adverse Network Environment}, 
  year={2023},
  volume={19},
  number={11},
  pages={11091-11101},
  doi={10.1109/TII.2022.3232772}
}

@ARTICLE{Huo2023,
  author={Huo, Xiang and Liu, Mingxi},
  journal={IEEE Transactions on Industrial Informatics}, 
  title={Encrypted Decentralized Multi-Agent Optimization for Privacy Preservation in Cyber-Physical Systems}, 
  year={2023},
  volume={19},
  number={1},
  pages={750-761},
  doi={10.1109/TII.2021.3132940}
}

@ARTICLE{Marcantoni2022,
  author={Marcantoni, Matteo and Jayawardhana, Bayu and Chaher, Mariano Perez and Bunte, Kerstin},
  journal={IEEE Control Systems Letters}, 
  title={Secure Formation Control via Edge Computing Enabled by Fully Homomorphic Encryption and Mixed Uniform-Logarithmic Quantization}, 
  year={2023},
  volume={7},
  pages={395-400},
  doi={10.1109/LCSYS.2022.3188944}
}

@inproceedings{Paillier1999,
  author = {Paillier, P.},
  title = {Public-key cryptosystems based on composite degree residuosity classes},
  booktitle = {ADVANCES IN CRYPTOLOGY - EUROCRYPT'99},
  series = {Lecture Notes in Computer Science},
  year = {1999},
  volume = {1592},
  pages = {223-238}
}

@article{SB23,
  author = {N. Schl\"{u}ter and P. Binfet and M. {Schulze Darup}},
  title = {A brief survey on encrypted control: From the first to the second generation and beyond},
  journal = {Annual Reviews in Control},
  volume = {56},
  pages = {100913},
  year = {2023},
  issn = {1367-5788},
  doi = {https://doi.org/10.1016/j.arcontrol.2023.100913}
}

@article{Darup2018,
  author = {Darup, M. S. and Redder, A. and Shames, I. and Farokhi, F. and Quevedo, D.},
  journal = {IEEE Control Systems Letters}, 
  title = {Towards Encrypted {MPC} for Linear Constrained Systems}, 
  year = {2018},
  volume = {2},
  number = {2},
  pages = {195-200},
  doi = {10.1109/LCSYS.2017.2779473}
}

@inproceedings{Alexandru2018,
	author = {Alexandru, A. B. and Morari, M. and Pappas, G. J.},
	booktitle = {2018 IEEE Conference on Decision and Control (CDC)}, 
	title = {Cloud-Based {MPC} with Encrypted Data}, 
	year = {2018},
	volume = {},
	number = {},
	pages = {5014-5019},
	doi = {10.1109/CDC.2018.8619835}
}

@ARTICLE{HS22,	
	author={Hu, Zhongrui and Shi, Peng and Wu, Ligang},	
	journal={IEEE Transactions on Circuits and Systems II: Express Briefs}, 	
	title={Preserving State and Control Privacies in Networked Systems With Tokenized Polytopic Transforms}, 	
	year={2022},	
	volume={69},	
	number={1},	
	pages={104-108},	
	doi={10.1109/TCSII.2021.3075471}}

@INPROCEEDINGS{SD20,	
	author={Schlüter, Nils and Darup, Moritz Schulze},	
	booktitle={2020 59th IEEE Conference on Decision and Control (CDC)}, 	
	title={Encrypted explicit {MPC} based on two-party computation and convex controller decomposition}, 	
	year={2020},	
	volume={},	
	number={},	
	pages={5469-5476},	
	doi={10.1109/CDC42340.2020.9304078}}

@article{BGV2014,
	author = {Brakerski, Z. and Gentry, C. and Vaikuntanathan, V.},
	title = {(Leveled) Fully Homomorphic Encryption without Bootstrapping},
	year = {2014},
	publisher = {Association for Computing Machinery},
	volume = {6},
	number = {3},
	issn = {1942-3454},
	doi = {10.1145/2633600},
	journal = {ACM Trans. Comput. Theory}
	}

@book{borrelli_bemporad_morari_2017, 
	place={Cambridge}, 
	title={Predictive Control for Linear and Hybrid Systems}, 
	DOI={10.1017/9781139061759}, 
	publisher={Cambridge University Press}, 
	author={Borrelli, F. and Bemporad, A. and Morari, M.}, 
	year={2017}
}

@article{chen:2020:ijrnc,
author = {Chen, Jiming and Gupta, Vijay and Quevedo, Daniel E. and Tesi, Pietro},
title = {Privacy and security of cyberphysical systems},
journal = {International Journal of Robust and Nonlinear Control},
volume = {30},
number = {11},
pages = {4165-4167},
doi = {https://doi.org/10.1002/rnc.5051},
url = {https://onlinelibrary.wiley.com/doi/abs/10.1002/rnc.5051},
eprint = {https://onlinelibrary.wiley.com/doi/pdf/10.1002/rnc.5051},
year = {2020}
}

@article{lin:2022:ijrnc,
year = {2025},
author = {Lin, Wenhao and Ni, Yuqing and Ren, Xiaoqiang and Yang, Wen and Yang, Chao},
title = {Privacy Preservation by Public Design in Cloud-Based Cooperative LQG Control Systems},
journal = {International Journal of Robust and Nonlinear Control},
volume = {n/a},
number = {n/a},
pages = {},
doi = {https://doi.org/10.1002/rnc.7778},
url = {https://onlinelibrary.wiley.com/doi/abs/10.1002/rnc.7778},
eprint = {https://onlinelibrary.wiley.com/doi/pdf/10.1002/rnc.7778}
}

@InProceedings{yalmip:paper,
  author       = {J. L{\"o}fberg},
  title	       = {{YALMIP : A Toolbox for Modeling and Optimization in
                  MATLAB}},
  booktitle    = {Proc.~of the {CACSD} Conference},
  year	       = {2004},
  address      = {Taipei, Taiwan},
  note	       = {}
}

@InProceedings{RLWE,
	author={Lyubashevsky, V. and Peikert, C. and Regev, O.},
	title={On Ideal Lattices and Learning with Errors over Rings},
	booktitle={Advances in Cryptology -- EUROCRYPT 2010},
	year={2010},
	publisher={Springer Berlin Heidelberg},
	address={Berlin, Heidelberg},
	pages={1--23},
	isbn={978-3-642-13190-5}
	}

@article{MR00,
	author = {D.Q. Mayne and J.B. Rawlings and C.V. Rao and P.O.M. Scokaert},
	title = {Constrained model predictive control: Stability and optimality},
	journal = {Automatica},
	volume = {36},
	number = {6},
	pages = {789-814},
	year = {2000},
	issn = {0005-1098},
	doi = {https://doi.org/10.1016/S0005-1098(99)00214-9}
}

@article{uiam1171,
	author	=	{M. Kvasnica and J. L\"ofberg and M. Fikar},
	title	=	{Stabilizing polynomial approximation of explicit {MPC}},
	journal	=	{Automatica},
	year	=	{2011},
	volume	=	{47},
	number	=	{10},
	pages	=	{2292-2297},
	doi	=	{10.1016/j.automatica.2011.08.023},
	url	=	{https://www.uiam.sk/assets/publication_info.php?id_pub=1171}
}

@ARTICLE{4026638,
	author={Baoti, M. and Christophersen, F. J. and Morari, M.},
	journal={IEEE Transactions on Automatic Control}, 
	title={Constrained Optimal Control of Hybrid Systems With a Linear Performance Index}, 
	year={2006},
	volume={51},
	number={12},
	pages={1903-1919},
	doi={10.1109/TAC.2006.886486}}

@Article{MayEtal:aut:00,
  author       = {D. Q. Mayne and J. B. Rawlings and C. V. Rao and
                  P. O. M. Scokaert},
  title	       = {Constrained model predictive control: Stability and
                  optimality},
  journal      = {Automatica},
  year	       = {2000},
  volume       = {36},
  number       = {6},
  pages	       = {789--814},
  month	       = jun,
  type	       = {Survey Paper},
}

@article{MS05,
	author = {D.Q. Mayne and M.M. Seron and S.V. Raković},
	title = {Robust model predictive control of constrained linear systems with bounded disturbances},
	journal = {Automatica},
	volume = {41},
	number = {2},
	pages = {219-224},
	year = {2005},
	issn = {0005-1098},
	doi = {https://doi.org/10.1016/j.automatica.2004.08.019}
}

@book{christophersen2007optimal,
	title={Optimal Control of Constrained Piecewise Affine Systems},
	author={Christophersen, F.},
	isbn={9783540727019},
	series={Lecture Notes in Control and Information Sciences},
	year={2007},
	publisher={Springer Verlag}
}

@article{MOK2008524,
	title = {Effective Pólya semi-positivity for non-negative polynomials on the simplex},
	journal = {Journal of Complexity},
	volume = {24},
	number = {4},
	pages = {524-544},
	year = {2008},
	issn = {0885-064X},
	doi = {https://doi.org/10.1016/j.jco.2008.01.003},
	url = {https://www.sciencedirect.com/science/article/pii/S0885064X0800006X},
	author = {H. N. Mok and W. K. To}
}

@inproceedings{kaluz:2019:flexy,
author	=	{M. Kal\'uz and M. Klau\v{c}o and {\v{L}}. {\v{C}}irka and M. Fikar},
title	=	{Flexy2: A Portable Laboratory Device for Control Engineering Education},
booktitle	=	{12th IFAC Symposium Advances in Control Education},
year	=	{2019},
pages	=	{159-164},
month	=	{July 7-9},
url	=	{https://www.uiam.sk/assets/publication_info.php?id_pub=2068}
}

@inproceedings{Flexy2,
author	=	{M. Kal\'uz and M. Klau\v{c}o and {\v{L}}. {\v{C}}irka and M. Fikar},
title	=	{Flexy2: A Portable Laboratory Device for Control Engineering Education},
booktitle	=	{12th IFAC Symposium Advances in Control Education},
year	=	{2019},
pages	=	{159-164},
month	=	{July 7-9},
url	=	{https://www.uiam.sk/assets/publication_info.php?id_pub=2068}
}

@misc{tenseal2021,
    title={TenSEAL: A Library for Encrypted Tensor Operations Using Homomorphic Encryption}, 
    author={A. Benaissa and B. Retiat and B. Cebere and A. E. Belfedhal},
    year={2021},
    eprint={2104.03152},
    archivePrefix={arXiv},
    primaryClass={cs.CR}
}

@misc{sealcrypto,
        title = {{M}icrosoft {SEAL} (release 4.0)},
        howpublished = {\url{https://github.com/Microsoft/SEAL}},
        month = mar,
        year = 2022,
        note = {Microsoft Research, Redmond, WA.},
        key = {SEAL}
    }

@InProceedings{kim2021,
author="Kim, A.
and Polyakov, Y.
and Zucca, V.",
title="Revisiting Homomorphic Encryption Schemes for Finite Fields",
booktitle="Advances in Cryptology -- ASIACRYPT 2021",
year="2021",
publisher="Springer International Publishing",
address="Cham",
pages="608--639",
isbn="978-3-030-92078-4"
}

@ARTICLE{MoritzDarup_2025_polyEcryption,
	author={Teichrib, Dieter and Adamek, Janis and Binfet, Philipp and Darup, Moritz Schulze},
	journal={IEEE Access}, 
	title={Polynomial Function Approximations With Leading Integer Coefficients for Efficient Encrypted Implementations}, 
	year={2025},
	volume={13},
	number={},
	pages={157455-157462},
	doi={10.1109/ACCESS.2025.3606013}}

@article{BFV_stateControl2024,
	title = {Optimizing encrypted control algorithms for real-time secure control},
	journal = {Journal of the Franklin Institute},
	volume = {361},
	number = {5},
	pages = {106677},
	year = {2024},
	issn = {0016-0032},
	doi = {https://doi.org/10.1016/j.jfranklin.2024.106677},
	url = {https://www.sciencedirect.com/science/article/pii/S001600322400098X},
	author = {Tongtong Sui and Jizhi Wang and Wen Liu and Jingshan Pan and Lizhen Wang and Yue Zhao and Lingrui Kong}
}

@ARTICLE{NonlinearControl_Survey_2026,
	author={Bian, Song and Jin, Yuexiang and Zhao, Dong and Fu, Yunhao and Pan, Haowen and Chen, Yi and Zhang, Bo and Ren, Changrui and Yin, Peng and Dong, Jin and Guan, Zhenyu},
	journal={IEEE Transactions on Information Forensics and Security}, 
	title={ENClose: Encrypted Nonlinear Closed-Loop Control Over Fully Homomorphic Encryption}, 
	year={2026},
	volume={21},
	number={},
	pages={3928-3943},
	doi={10.1109/TIFS.2026.3683287}}

@ARTICLE{EncryptedRobustMPC2026,
	author={Peng, Kai-Yu and Xie, Wei and Zhang, Langwen and Mo, Wen-jun},
	journal={IEEE Transactions on Control of Network Systems}, 
	title={Homomorphically Encrypted Robust MPC for Privacy-Preserving Control Under System Uncertainties}, 
	year={2026},
	volume={13},
	number={2},
	pages={923-934},
	doi={10.1109/TCNS.2026.3667762}}

@article{EncryptedDistributedMPC2024,
	title = {Encrypted distributed model predictive control of nonlinear processes},
	journal = {Control Engineering Practice},
	volume = {145},
	pages = {105874},
	year = {2024},
	issn = {0967-0661},
	doi = {https://doi.org/10.1016/j.conengprac.2024.105874},
	url = {https://www.sciencedirect.com/science/article/pii/S0967066124000340},
	author = {Yash A. Kadakia and Fahim Abdullah and Aisha Alnajdi and Panagiotis D. Christofides}
}

\end{document}